\documentclass[11pt]{article}
\usepackage{amsmath,amssymb}
\usepackage{amsthm} 
\usepackage{tikz}
\usetikzlibrary{arrows}
\usetikzlibrary{intersections}
\usepackage{tkz-euclide}
\usetkzobj{all}
\usepackage{latexsym}
\newtheorem{thm}{Theorem}[section]
\newtheorem{lem}[thm]{Lemma}
\newtheorem{cor}[thm]{Corollary}
\newtheorem{prop}[thm]{Proposition}

\begin{document}
\title{\bf An improved algorithm for recognizing matroids}
\date{}
\maketitle
\begin{center}
\author{
{\bf Brahim Chaourar} \\
{ Department of Mathematics and Statistics,\\Imam Mohammad Ibn Saud  Islamic University (IMSIU) \\P.O. Box
90950, Riyadh 11623,  Saudi Arabia }\\{Correspondence address: P. O. Box 287574, Riyadh 11323, Saudi Arabia}\\{email: bchaourar@hotmail.com}}
\end{center}



\begin{abstract}
\noindent Let $M$ be a matroid defined on a finite set $E$ and $L\subset  E$. $L$ is locked in $M$ if $M|L$ and $M^*|(E\backslash L)$ are 2-connected, and $min\{r(L), r^*(E\backslash L)\} \geq 2$. Locked subsets characterize nontrivial facets of the bases polytope. In this paper, we give a new axiom system for matroids based on locked subsets. We deduce an algorithm for recognizing matroids improving the running time complexity of the best known till today. This algorithm induces a polynomial time algorithm for recognizing uniform matroids. This latter problem is intractable if we use an independence oracle.
\end{abstract}


\section{Introduction}

Sets and their characteristic vectors will not be distinguished. We refer to Oxley \cite{Oxley 1992} and Schrijver \cite{Schrijver 1986} about matroids and polyhedra terminology and facts, respectively.
\newline Given a matroid $M$ defined on a finite set $E$. Suppose that $M$ (and $M^*$) is 2-connected. A subset $L\subset E$ is called a locked subset of $M$ if $M|L$ and $M^*|(E\backslash L)$ are 2-connected, and their corresponding ranks are at least 2, i.e., $min\{r(L), r^*(E\backslash L)\} \geq 2$. It is not difficult to see that if $L$ is locked then both $L$ and $E\backslash L$ are closed, respectively, in $M$ and $M^*$ (That is why we call it locked). We denote by $\mathcal{L}(M)$ and $\ell(M)$, respectively, the class of locked subsets of $M$ and its cardinality, which is called the locked number of $M$. For a disconnected matroid $M$, it is not difficult to see that the class of locked subsets of $M$ is the union of locked subsets of the 2-connected components of $M$. The locked structure of $M$ is the quadruple ($\mathcal P(M)$, $\mathcal S(M)$, $\mathcal L(M)$, $r$ ), where  $\mathcal P(M)$ and $\mathcal S(M)$ are, respectively, the class of parallel and coparallel closures, and $r$ is the rank function restricted to $\mathcal P(M)\cup \mathcal S(M)\cup \mathcal L(M)\cup \{\varnothing , E\}$. For $x\in \mathbb{R}^E$ and $Y\subseteq E$, $x(Y)=\sum_{e\in Y} x(e)$. For any class $\mathcal X\subseteq 2^E$, $\mathcal X^C=\{ E\backslash X$ such that $X\in \mathcal X\}$. We use the notations: $r(M)=r(E)$ and $r^*(M^*)=r^*(E)$.
\newline A matroid $M$ can be completely characterized by its locked structure through its bases polytope $BP(M)$ \cite{Chaourar 2018}:
\begin{thm} A minimal description of $BP(M)$ is the set of all $x\in \mathbb{R}^E$ satisfying:
$$     x(E)=r(E)                                                                    \eqno       $$
$$ x(P) \leq 1 \> \mathrm{for\>any\> parallel\> closure}\> P\subseteq E   \eqno  $$
$$ x(S) \geq |S|-–1  \> \mathrm{for\> any\> coparallel\> closure}\> S\subseteq E \eqno     $$
$$ x(L) \leq r(L)  \> \mathrm{for\> any\> locked\> subset}\> L\subseteq E   \eqno     $$
\end{thm}
This reflects the importance of the locked structure of a matroid. We give in this paper a new axiom system for defining matroids based on this quadruple. We deduce an improved algorithm for recognizing matroids. This problem is intractable (see \cite{Robinson and Welsh 1980}). A similar study has been done by Provan and Ball for testing if a given clutter $\Omega$, defined on a finite set $E$, is the class of the bases of a matroid \cite{Provan and Ball 1988}. They provide an algorithm with running time complexity $O(|\Omega|^3 |E|)$. Spinrad \cite{Spinrad 1991} improves the running time to $O(|\Omega|^2 |E|)$. In this paper, we give an algorithm for matroid recognition with running time complexity $O(|E|^2+|\mathcal L|^2+|E||\mathcal L|log|\mathcal L|)$. This improves the running time complexity of Spinrad's algorithm. Our algorithm becomes polynomial on $|E|$ for recognizing uniform matroids. Recognizing uniform matroids is intractable if we use an independence or an equivalent oracle \cite{Jensen and Korte 1982}.
\newline The remainder of this paper is organized as follows. In section 2, we give a new axiom system for defining matroids; then, in section 3, we give an improved algorithm for recognizing matroids which induces a polynomial time algorithm for recognizing uniform matroids. Finally, we will conclude in section 4.

\section{The Locked Axioms for a Matroid}

We will give a list of axioms for defining a matroid. Even if this list seems to be long but it simplifies many problems, for example recognizing matroids. Moreover, on the inverse of known axioms defining matroids, most of our axioms can be verified efficiently.
\newline Given a finite set $E$, the basic quadruple $M=(\mathcal P, \mathcal S, \mathcal L, r)$ is a locked system defined on $E$ if it verifies the followings:
\newline (L1) $E\neq \varnothing$,
\newline (L2) $\mathcal P$ and $\mathcal S$ are partitions of $E$,
\newline (L3) For any $(P, S)\in \mathcal P\times \mathcal S$, if $P\cap S\neq \varnothing$ then $|P|=1$ or $|S|=1$,
\newline (L4) $\mathcal L$ is a class of nonempty and proper subsets of $E$ such that $\mathcal L\cap \mathcal P=\mathcal L\cap \mathcal S=\varnothing$,
\newline (L5) For any $(X, L)\in (\mathcal P\cup \mathcal S)\times \mathcal L$, $X\cap L= \varnothing$ or $X\subset L$,
\newline (L6) $r$ is a nonegative function defined on $2^E$,
\newline (L7) $r(\varnothing)=0$ and $r(E)\geq r(X)$ for any $X\subseteq E$,
\newline (L8) $r(P)=min\{1, r(E)\}$ for any $P\in \mathcal P$,
\newline (L9) $r(E\backslash P)=min\{|E\backslash P|, r(E)\}$ for any $P\in \mathcal P$,
\newline (L10) $r(S)=min\{|S|, r(E)\}$ for any $S\in \mathcal S$,
\newline (L11) $r(E\backslash S)=min\{|E\backslash S|, r(E)+1-|S|\}$ for any $S\in \mathcal S$,
\newline (L12) $r(L)\geq max\{2, r(E)+2-|E\backslash L|\}$ for any $L\in \mathcal L$,
\newline (L13) $r$ is increasing on $\mathcal P\cup \mathcal L\cup \{\varnothing, E\}$, i.e., for any $(X, Y)\in (\mathcal P\cup \mathcal L\cup \{\varnothing, E\})^2$, $r(X)<r(Y)$ if $X\subset Y$,
\newline (L14) $r$ is submodular on $\mathcal P\cup \mathcal S\cup \mathcal L\cup \{\varnothing, E\}$, i.e., for any $(X, Y)\in (\mathcal P\cup \mathcal S\cup \mathcal L\cup \{\varnothing, E\})^2$, $r(X\cap Y)+r(X\cup Y)\leq r(X)+r(Y)$.
\newline (L15) For any $X\not \in \mathcal P\cup \mathcal P^C\cup \mathcal S\cup \mathcal S^C\cup \mathcal L\cup \{\varnothing, E\}$, one of the following holds (recursive property):
\begin{list}{}
{\leftmargin=0.5in \itemindent=0cm}
\item (P0) There exists $Y\in \mathcal P\cup \mathcal S$ such that $X\subset Y$ , $r(X)=1$ if $Y\in \mathcal P$ and $r(X)=|X|$ otherwise,
\item (P1) There exists $L\in \mathcal L$ such that $L\subset X$ , $r(X)=r(L)+r(X\backslash L)$, and $X\backslash L$ verifies (P1) or (P2),
\item (P2) There exists $P\in \mathcal P$ such that $P\cap X\neq \varnothing$ , $r(X)=r(P)+r(X\backslash P)$, and $X\backslash P$ verifies (P1) or (P2),
\item (P3) There exists $L\in \mathcal L$ such that $X\subset L$ , $r(X)=r(L)+r(X\cup (E\backslash L))-r(E)$, and $X\cup (E\backslash L)$ verifies (P3) or (P4),
\item (P4) There exists $S\in \mathcal S$ such that $(E\backslash S)\cup X\neq E$ , $r(X)=r(E\backslash S)+r(X\cup S)+|S\cap X|-r(E)$, and $X\cup S$ verifies (P3) or (P4),
\end{list}
(L16) For any $(L_1, L_2)\in \mathcal L^2$, if $L_1\cap L_2\neq \varnothing$ and $L_1\cap L_2\not \in \mathcal L$ then $L_1\cap L_2$ verifies (P0), (P1) or (P2) of (L15),
\newline (L17) For any $(L_1, L_2)\in \mathcal L^2$, if $L_1\cup L_2\neq E$ and $L_1\cup L_2\not \in \mathcal L$ then $L_1\cup L_2$ verifies (P3) or (P4) of (L15),
\newline Without loss of generality, we can replace axioms (L8)-(L11) by the following axioms respectively:
\newline (LL8) $r(P)=1$ for any $P\in \mathcal P$,
\newline (LL9) $r(E\backslash P)=r(E)$ for any $P\in \mathcal P$,
\newline (LL10) $r(S)=|S|$ for any $S\in \mathcal S$,
\newline (LL11) $r(E\backslash S)=r(E)+1-|S|$ for any $S\in \mathcal S$.
\newline Let us give the following polyhedron associated to the locked system $M$:
\newline  $P(M)$ is the set of all $x\in R^E$ satisfying the following constraints:
$$ x(E)=r(E)                                                                      \eqno  (1)$$
$$ x(P) \leq 1 \> \> \mathrm{for\>any}\> P\in \mathcal P   \eqno  (2)$$
$$ x(S) \geq |S|-–1  \> \mathrm{for\>any}\> S\in \mathcal S                    \eqno  (3)$$
$$ x(L) \leq r(L)  \> \mathrm{for\> any}\> L\in \mathcal L   \eqno  (4)$$
Now, we can start our process to prove the main theorem.
\begin{lem}
If $x\in P(M)$ then $$0\leq x(e) \leq 1 \> \> \mathrm{for\>any}\>e\in E \eqno  (5)$$
\end{lem}
\begin{proof} Let $e\in E$. Since $\mathcal P$ and $\mathcal S$ are partitions of $E$ (L2) then there exist a pair $(P, S)\in \mathcal P\times\mathcal S$ such that $\{e\}=P\cap S$ (L3).
\newline \textbf{Case 1:} if $|P|=|S|=1$ then inequalities (2) and (3) imply inequalties (5).
\newline \textbf{Case 2:} if $|S|\geq 2$ then $\{f\}\in \mathcal P$ for any $f\in S$ (L3). Inequalities (2) imply $x(f)\leq 1$ for any $f\in S$. In particular  $x(e)\leq 1$. It follows that $x(S\backslash \{e\})\leq |S|-1$, then $x(e)=x(S)-x(S\backslash \{e\})\geq (|S|-1)-(|S|-1)=0$
\newline \textbf{Case 3:} if $|P|\geq 2$ then $\{f\}\in \mathcal S$ for any $f\in P$ (L3). Inequalities (3) imply $x(f)\geq 0$ for any $f\in P$. In particular  $x(e)\geq 0$. It follows that $x(P\backslash \{e\})\geq 0$, then $x(e)=x(P)-x(P\backslash \{e\})\leq x(P)\leq 1$.
\end{proof}
\begin{lem}
If $x\in P(M)$ then $$ x(A) \leq r(A) \> \> \mathrm{for\>any}\>A\subseteq E \eqno  (6)$$
\end{lem}
\begin{proof} We have the following cases:
\newline \textbf{Case 1:} if $A=\varnothing$ then $x(A)=0\leq 0=r(A)$ (L7).
\newline \textbf{Case 2:} if $A=E$ then $x(A)=r(A)\leq r(A)$ (inequality 1).
\newline \textbf{Case 3:} if $A\in \mathcal P$ then $x(A) \leq 1=r(A)$ (inequality 2 and LL8).
\newline \textbf{Case 4:} if $A\in \mathcal S$ then $x(A)\leq |A|=r(A)$ (Lemma 2.1 and LL10).
\newline \textbf{Case 5:} if $E\backslash A\in \mathcal S$ then $x(A)=x(E)-x(E\backslash A)\leq r(E)-|E\backslash A|+1=r(A)$ (inequality 3 and LL11).
\newline \textbf{Case 6:} if $A\not \in \mathcal P\cup \mathcal S\cup \mathcal L\cup \{\varnothing, E\}$ and $E\backslash A\not \in \mathcal S$ then the axiom (L15) implies one of the following subcases:
\newline \textbf{Subcase 6.0:} There exists $Y\in \mathcal P\cup \mathcal S$ such that $A\subset Y$, then $r(A)=1=r(Y)\geq x(Y)\geq x(A)$ (Case 3) or $r(A)=|A|\geq x(A)$ (inequality 5).
\newline \textbf{Subcase 6.1:} There exists $L\in \mathcal L$ such that $L\subset A$ , $r(A)=r(L)+r(A\backslash L)$, and $A\backslash L$ verifies (P1) or (P2). So by induction on $|A|$, $x(A)=x(L)+x(A\backslash L)\leq r(L)+r(A\backslash L)=r(A)$ because $|A\backslash L|<|A|$ and inequality 4.
\newline \textbf{Subcase 6.2:} There exists $P\in \mathcal P$ such that $P\cap A\neq \varnothing$ , $r(A)=r(P)+r(A\backslash P)$, and $A\backslash P$ verifies (P1) or (P2). So by induction on $|A|$, $x(A)\leq x(P)+x(A\backslash P)\leq r(P)+r(A\backslash P)=r(A)$ because $|A\backslash P|<|A|$, Lemma 2.1 and Case 3,
\newline \textbf{Subcase 6.3:} There exists $L\in \mathcal L$ such that $A\subset L$ , $r(A)=r(L)+r(A\cup (E\backslash L))-r(E)$, and $A\cup (E\backslash L)$ verifies (P3) or (P4). So by induction on $|E\backslash A|$, $x(A)=x(L)+x(A\cup (E\backslash L))-x(E)\leq r(L)+r(A\cup (E\backslash L))-r(E)=r(A)$ because $|E\backslash(A\cup (E\backslash L))|=|(E\backslash A)\cap L|<|E\backslash A|$ and inequality 4,
\newline \textbf{Subcase 6.4:} There exists $S\in \mathcal S$ such that $(E\backslash S)\cup A\neq E$ , $r(A)=r(E\backslash S)+r(A\cup S)+|S\cap A|-r(E)$, and $A\cup S$ verifies (P3) or (P4). So by induction on $|E\backslash A|$,
$$x(A)=x(E\backslash S)+x(A\cup S)+x(S\cap A)-x(E)\leq r(E\backslash S)+r(A\cup S)+|S\cap A|-r(E)=r(A)$$
because $|E\backslash(A\cup S)|=|(E\backslash A)\cap (E\backslash S)|<|E\backslash A|$, Lemma 2.1, Case 5, and inequality 4.
\end{proof}
 Let $Q(M)$ be the set of $x\in R^E$ such that $x$ verifies the inequalities (1), (5) and (6).
\begin{cor} $P(M)=Q(M)$.
\end{cor}
\begin{proof} Lemma 2.1 and 2.2 imply that $P(M)\subseteq Q(M)$. We need to prove the inverse inclusion. Let $x\in Q(M)$. It is clear that $x$ verifies the inequalities (2) and (4) by using inequality (6) and axiom (LL8). Let $S\in \mathcal S$ then, by using inequalities (1), (6) and axiom (LL11), $x(S)=x(E)-x(E\backslash S)\geq r(E)-r(E\backslash S)=r(E)-r(E)-1+|S|=|S|-1$, which is inequality (3).
\end{proof}
\begin{lem} Let $x\in P(M)$ such that $x(L_i)=r(L_i)$, for some $L_i\in \mathcal L, i=1, 2$.
\newline If  $L_1\cap L_2\neq \varnothing$ then there exists $L\in\mathcal P\cup \mathcal L$ such that $L\subseteq L_1\cap L_2$ and $x(L)=r(L)$.
\end{lem}
\begin{proof} By using Lemma 2.2 and axiom (L14), we have:
$$ r(L_1)+r(L_2)=x(L_1)+x(L_2)=x(L_1\cap L_2)+x(L_1\cup L_2)\leq r(L_1\cap L_2)+r(L_1\cup L_2)\leq r(L_1)+r(L_2).$$
It follows that $x(L_1\cap L_2)=r(L_1\cap L_2)$ and $x(L_1\cup L_2)=r(L_1\cup L_2)$.
\newline If $L_1\cap L_2\in \mathcal L$ then $L=L_1\cap L_2$.
\newline Otherwise, by using axiom (L16), we have two cases:
\newline \textbf{Case 1:} There exists $L\in \mathcal L$ such that $L\subset L_1\cap L_2$ , $r(L_1\cap L_2)=r(L)+r((L_1\cap L_2)\backslash L)$, and $(L_1\cap L_2)\backslash L$ verifies (P1) or (P2) of axiom (L15). It is not difficult to see, by a similar argument as hereinabove, that $x(L)=r(L)$ and $x((L_1\cap L_2)\backslash L)=r((L_1\cap L_2)\backslash L)$.
\newline \textbf{Case 2:} There exists $P\in \mathcal P$ such that $P\cap (L_1\cap L_2)\neq \varnothing$ , $r(L_1\cap L_2)=r(P)+r((L_1\cap L_2)\backslash P)$, and $(L_1\cap L_2)\backslash P$ verifies (P1) or (P2) of axiom (L15). Axiom (L5) implies that $P\subseteq L_1\cap L_2$. It is not difficult to see, by a similar argument as hereinabove, that $x(P)=r(P)$ and $x((L_1\cap L_2)\backslash P)=r((L_1\cap L_2)\backslash P)$.
\end{proof}
\begin{thm} $P(M)$ is integral. \end{thm}
\begin{proof} Let $x\in P(M)$ be a fractional extreme point and $F=\{g\in E$ such that $0<x(g)<1\}$. Since $x$ is fractional and $x(E)=r(E)$ is integral then $|F|\geq 2$.
\newline Let $\mathcal P_x=\{P\in \mathcal P$ such that $x(P)=1\}$, $\mathcal S_x=\{S\in \mathcal S$ such that $x(S)=|S|-1\}$, and $\mathcal L_x=\{L\in \mathcal L$ such that $x(L)=r(L)\}$, i.e., the corresponding tight constraints of $x$.
\newline \textbf{Case 1:} There exists $X\in \mathcal P\cup \mathcal S$ such that $|X\cap F|\geq 2$. Let $\{e, f\}\subseteq X\cap F$. It follows that there exists $\varepsilon >0$ such that $0<x(e)-\varepsilon <1$ and $0<x(f)+\varepsilon <1$. Let $x_\varepsilon \in R^E$ such that:
\[ x_\varepsilon (g)= \left\{ \begin{array}{ll}
   x(g) & \mbox{if $g\not \in \{e, f\}$};\\
   x(e)-\varepsilon & \mbox{if $g=e$ };\\
   x(f)+\varepsilon & \mbox{if $g=f$}.
   \end{array} \right.\]
It is clear that $x_\epsilon(E)=r(E)$. Axioms (L2), (L3) and (L5) imply that $\mathcal P_x=\mathcal P_{x_\varepsilon}$, $\mathcal S_x=\mathcal S_{x_\varepsilon}$, and $\mathcal L_x=\mathcal L_{x_\varepsilon}$, i.e., $x_\varepsilon$ verifies the same tight constraints as $x$, a contradiction.
\newline \textbf{Case 2:} For any $X\in \mathcal P\cup \mathcal S$, we have $|X\cap F|\leq 1$. It follows that for any  $X\in \mathcal P_x\cup \mathcal S_x$, we have $X\cap F=\varnothing$.
\newline \textbf{Subcase 2.1:} There exists $L\in \mathcal L_x$ such that $|L\cap F|\geq 2$, and $\{ e, f\}\subseteq L\cap F$ such that if $L'\in \mathcal L_x$ then $\{ e, f\}\subseteq L'$ or $\{ e, f\}\cap L'=\varnothing$. So we proceed as in Case 1 and we conclude.
\newline \textbf{Subcase 2.2:} For any $L\in \mathcal L_x$ such that $|L\cap F|\geq 2$, and any $\{ e, f\}\subseteq L\cap F$, there exists $L'\in \mathcal L_x$ such that $|\{ e, f\}\cap L'|=1$. Suppose that $f\in L'$. So we have:
$$ r(L)+r(L')=x(L)+x(L')=x(L\cap L')+x(L\cup L')\leq r(L\cap L')+r(L\cup L')\leq r(L)+r(L').$$
It follows that $x(L\cap L')=r(L\cap L')$ and $x(L\cup L')=r(L\cup L')$. It follows that $|L\cap L'\cap F|\geq 2$.
\newline By using Lemma 2.4, and since $L\cap L'\neq \varnothing$ then there exists $X\in \mathcal P\cup \mathcal L$ such that $X\subseteq L\cap L'$ and $x(X)=r(X)$. By induction on $|L\cap L'|$, we have $X\cap F\neq \varnothing$ (otherwise we do the same for $(L\cap L')\backslash X$), i.e.,  $|X\cap F|\geq 2$. Induction on $|L|$ and axiom (L13) imply that $r(X)=1$, i.e., $X\in \mathcal P$, a contradiction.
\end{proof}
Now we can state our main theorem as follows.
\begin{thm} The extreme points of $P(M)$ are the bases of a matroid defined on $E$, and $\mathcal P, \mathcal S, \mathcal L, r$ are, respectively, the class of parallel and coparallel closures, locked subsets and rank function of this matroid.
\end{thm}
\begin{proof}
Lemma 2.1 and Theorem 2.5 imply that the extreme points of $P(M)$ are in $\{0, 1\}^E$. We remind here that we will not distinguish between sets and $\{0, 1\}$-vectors.
\newline Constraint (1) implies that extreme points of $P(M)$ have the same cardinality $r(E)$. We only need to prove the basis exchange axiom. We will do it by contradiction.
\newline Let $x$ and $x'$ be two extreme points of $P(M)$ and $e\in x\backslash x'$ such that for any $f\in x'\backslash x$, $x-e+f$ is not an extreme point, i.e., $x-e+f\not \in P(M)$. It is clear that $|x'\backslash x|\geq 2$. Let $x_f=x-e+f$.
\newline \textbf{Case 1:} $x_f$ violates an inequality of type (2), i.e., there exists $P_f\in \mathcal P$ such that $x_f(P_f)\geq 2$. It follows that $e\not \in P_f$, $f\in P_f$, $x(P_f)=1$, and $x_f(P_f)=2$. Thus there exists $f'\in P_f\cap x_f\cap x$ such that $f'\neq f$.
\newline \textbf{Claim:} If $f_1\neq f_2$ then $f'_1\neq f'_2$.
\newline Suppose, by contradiction, that $f'=f'_1=f'_2$. Since $f'\in P_{f_i}\cap x_{f_i}, i=1, 2$, then $f'\in P_{f_1}\cap P_{f_2}$. Axiom (L2) implies that $P_{f_1}=P_{f_2}=P$ and $\{f_1, f_2\}\subseteq P\cap x'$. It follows that $x'(P)\geq 2$, a contradiction.
\newline Since $|x\backslash x'|=|x'\backslash x|$ then $x\backslash x'=\bigcup \limits_{i=1}^{|x\backslash x'|} \{f'_i\}\subseteq \bigcup \limits_{i=1}^{|x\backslash x'|} P_{f_i}$ but $e\not \in P_{f_i}, i=1, 2, ..., |x\backslash x'|$, a contradiction.
\newline \textbf{Case 2:} $x_f$ violates an inequality of type (3), i.e., there exists $S_f\in \mathcal S$ such that $x_f(S_f)\leq |S|-2$. It follows that $e\in S_f$, $f\not \in S_f$, $x(S_f)=|S_f|-1$, and $x_f(S_f)=|S_f|-2$. Since $e\in S_f$ for any $f\in x'\backslash x$, and by using axiom (L2), we have $S_f=S$, i.e., for distinct $f_1$ and $f_2$, $S_{f_1}=S_{f_2}=S$. It follows that $(x'\backslash x)\cap S=\varnothing$. But $x'(S)\geq |S|-1$ because $x'\in P(M)$, then $(x'\cap x)(S)\geq |S|-1$. It follows that $(x\backslash x')\cap S=\varnothing$, a contradiction with $e\in S$.
\newline \textbf{Case 3:}  $x_f$ violates an inequality of type (4), i.e., there exists $L_f\in \mathcal L$ such that $x_f(L_f)\geq r(L_f)+1$. It follows that $e\not \in L_f$, $f\in L_f$, $x(L_f)=r(L_f)$, and $x_f(L_f)=r(L_f)+1$. We choose $L_f$ maximal for this property.
\newline \textbf{Subcase 3.1:} There are $f_1\neq f_2$ such that $x_{f_1}(L_{f_2})=r(L_{f_2})$, i.e., $f_1\notin L_{f_2}$.
\newline As shown in the proof of Lemma 2.4, $x(L_{f_1}\cup (L_{f_2})=r(L_{f_1}\cup (L_{f_2})$. Since $L_{f_2}$ is maximal then $L_{f_1}\cup L_{f_2}\notin \mathcal L$. Since $e\not \in L_{f_1}\cup L_{f_2}$, and by using  axiom (L17), there exists $S\in \mathcal S$ such that $(E\backslash S)\cup (L_{f_1}\cup L_{f_2})\neq E$ , $r(L_{f_1}\cup L_{f_2})=r(E\backslash S)+r(L_{f_1}\cup L_{f_2}\cup S)+|S\cap (L_{f_1}\cup L_{f_2})|-r(E)$, and $(L_{f_1}\cup L_{f_2})\cup S$ verifies (P4) (property (P3) cannot be verified because of maximality of $L_{f_2}$). By a similar argument as in the proof of Lemma 2.4, we have:
\newline \textbf{(1)} $x(E\backslash S)=r(E\backslash S)$ which imply that $x(S)=|S|-1$, i.e., $S=x\backslash \{e'\}\cup \{ f'\}$ for some $e'\in x$ and $f'\notin x$,
\newline and  \textbf{(2)} $x(L_{f_1}\cup L_{f_2}\cup S)=r(L_{f_1}\cup L_{f_2}\cup S)$.
\newline If $e\in S$ (i.e. $e\neq e'$) then at least one the $x_{f_i}(S)=|S|-2$ (i.e. $f_i\neq f'$) and we are in Case 2).
Else $e\notin S$, i.e. $e\notin L_{f_1}\cup L_{f_2}\cup S$ and by induction on $|E\backslash X|$ where $X=L_{f_1}\cup L_{f_2}$, we get a contradiction.
\newline \textbf{Subcase 3.2:} For any $f_1\neq f_2$, $x_{f_1}(L_{f_2})=r(L_{f_2})+1$, i.e., $f_1\in L_{f_2}$. It follows that there exists $L\in \mathcal L$ such that $x'\backslash x\subseteq L$, $e\notin L$, and $x(L)=r(L)$. We have then:
\newline $r(L)\geq x'(L)=(x'\backslash x)(L)+(x'\cap x)(L)=|x'\backslash x|+(x'\cap x)(L)=|x\backslash x'|+(x'\cap x)(L)\geq (x\backslash x')(L)+(x'\cap x)(L)=x(L)=r(L)$. It follows that $(x\backslash x')(L)=|x\backslash x'|$, i.e., $x\backslash x'\subseteq L$, a contradiction with $e\in x\backslash x'$.
\end{proof}
Actually this gives a new proof for the bases polytope of a matroid and its facets based on the locked structure only.

\section{An improved algorithm for matroid recognition}

Since we have proved that the locked axioms define a matroid uniquely, thus, recognition of matroids is equivalent to recognize if a basic quadruple is a locked system. We give now the running time complexity for testing each of the needed locked axioms.
\newline (L1) can be tested in $O(1)$. (L2) can be tested in $O(|E|^2)$. (L3) can be tested in $O(|E|^2)$. (L4) and (L5) can be tested in $O(|E||\mathcal L|)$. We need the following lemma for (L6).
\begin{lem}
We can replace axiom (L6) by the following axiom:
\newline (LL6) $r$ is a nonnegative integer function defined on $\mathcal L\cup \{E\}$.
\end{lem}
\begin{proof}
Axioms (LL6) and (L7)-(L11) imply the following axiom:
\newline (LLL6) $r$ is a nonnegative integer function defined on $\mathcal P\cup \mathcal S\cup \mathcal L\cup \{\varnothing, E\}\cup \mathcal P^C\cup \mathcal S^C$, where $\mathcal X^C=\{E\backslash X$ such that $X\in \mathcal X\}$ and $\mathcal X\in \{\mathcal P, \mathcal S\}$.
\newline According to the coming Remark 1, it suffices to test (LLL6).
\end{proof}
It follows that (LL6) can be tested in $O(|\mathcal L|)$. We need the following lemma for (L7).
\begin{lem}
We can replace axiom (L7) by the following axiom:
\newline (LL7) $r(\varnothing)=0$.
\end{lem}
\begin{proof}
Axioms (LL7), (L13), (L14) and (L15) imply axiom (L7).
\end{proof}
It follows that (LL7) can be tested in $O(1)$. (L8)-(L11) can be tested in $O(|E|)$. (L12) can be tested in $O(|\mathcal L|)$. We need the following lemma for (L13).
\begin{lem}
(L13) can be tested in $O(|E||\mathcal L|log|\mathcal L|)$.
\end{lem}
\begin{proof}
We can construct a lattice (ordered by inclusion) for elements of $\mathcal P\cup \mathcal L\cup \{\varnothing, E\}$. The root is the empty set, and the sink is the ground set. Adjacent vertices to the root are the elements of $\mathcal P$ because of axioms (L4) and (L5). After sorting the elements of $\mathcal L$ according to their cardinalities, we can complete the lattice. We can test the axiom (L13) at each step of the lattice construction.
\end{proof}
(L14) can be tested in $O(|E|^2+|\mathcal L|^2)$.
\newline {\bf Remark 1:} Note that the axioms (L15)-(L17) give a way on how to compute the values of the function $r$ outside $\mathcal P\cup \mathcal P^C\cup \mathcal S\cup \mathcal S^C\cup \mathcal L\cup \{\varnothing, E\}$. So we do not need to verify them for a locked system realization. We need only to prove the following proposition.

\begin{prop}
For any $X\not \in \mathcal P\cup \mathcal P^C\cup \mathcal S\cup \mathcal S^C\cup \mathcal L\cup \{\varnothing, E\}$, one of the following holds:
\begin{list}{}
{\leftmargin=0.5in \itemindent=0cm}
\item (PP0) There exists $Y\in \mathcal P\cup \mathcal S$ such that $X\subset Y$,
\item (PP1) There exists $L\in \mathcal L$ such that $L\subset X$ , and $X\backslash L$ verifies (PP1) or (PP2),
\item (PP2) There exists $P\in \mathcal P$ such that $P\cap X\neq \varnothing$, and $X\backslash P$ verifies (PP2),
\item (PP3) There exists $L\in \mathcal L$ such that $X\subset L$, and $X\cup (E\backslash L)$ verifies (PP3) or (PP4),
\item (PP4) There exists $S\in \mathcal S$ such that $(E\backslash S)\cup X\neq E$, and $X\cup S$ verifies (PP4),
\end{list}
\end{prop}
\begin{proof}
If $X$ does not verify (PP0), (PP1), and (PP3), then, since $\mathcal P$ and $\mathcal S$ are partitions of $E$, $X$ should verify (PP2) or (PP4).
\end{proof}

We can summarize all the previous steps in our main result as follows.
\begin{thm}
We can decide if a basic quadruple $(\mathcal P, \mathcal S, \mathcal L, r)$ is a locked system or not in $O(|E|^2+|\mathcal L|^2+|E||\mathcal L|log|\mathcal L|)$.
\end{thm}
This algorithm improves the running time complexity of that given by Spinrad's algorithm \cite{Spinrad 1991} (testing if a given clutter forms the class of bases of  a matroid) because if the answer is yes (worst case running time), i.e., the given clutter form the class of bases of a matroid, then its running time complexity is $O(|\mathcal B|^2 |E|)$ where $\mathcal B$ is the class of bases, and $|\mathcal B|>|\mathcal P|+|\mathcal S|+|\mathcal L|>|E|+|\mathcal L|$ because the facets of the bases polytope are completely described by $\mathcal P\cup \mathcal S\cup \mathcal L\cup \{E\}$ (see Theorem 1.1) and the number of extreme points is greater than the number of facets. Furthermore, Spinrad's algorithm has in the input a clutter which is not a basic structure as in our algorithm (basic quadruple).
\newline A consequence of Theorem 3.5 is the following corollary about recognition of uniform matroids in polynomial time. We need the following theorem \cite{Chaourar 2018} for this purpose.
\begin{thm} A matroid $M$ is uniform if and only if one of the following properties holds:
\newline (i) $\ell (M)=0$ and $|\mathcal P(M)|=|E|=|\mathcal S(M)|$;
\newline (ii) $|\mathcal P(M)|=1$;
\newline (iii) $|\mathcal S(M)|=1$;
\newline (iv) $r(M)=|E|$;
\newline (v) $r(M)=0$.
\end{thm}
It follows that:
\begin{cor}
We can decide if a basic quadruple is an uniform matroid or not in $O(|E|^2)$.
\end{cor}
\begin{proof}
Testing (i) of Theorem 3.6 can be done in $O(|E|^2)$. (ii) of Theorem 3.6 is equivalent to: $|\mathcal P(M)|=1$, $|\mathcal S(M)|=|E|$ and $\ell (M)=0$. So we can test (ii) in $O(|E|^2)$. We can use a similar argument for (iii). (iv) of Theorem 3.6 is equivalent to: $r(M)=|E|$, $|\mathcal S(M)|=|E|=|\mathcal P(M)|$ and $\ell (M)=0$. So we can test (iv) in $O(|E|^2)$. Finally, a similar argument can be used for (v).
\end{proof}
We can present the latter algorithm in a different manner by using the following corollary of Theorem 3.6.
\begin{cor}
If $M$ is an uniform matroid then $\ell (M)=0$.
\end{cor}
Now we introduce a new matroid oracle.
\newline{}
\newline {\bf The zero locked number oracle}
\newline \begin{tabular}{*6l}
Input: & A finite set $E$ and a basic quadruple $M=(\mathcal P, \mathcal S, \mathcal L, r)$ defined on $E$. \\
Output:& (1) No if $\mathcal L\neq \O$ \\
       & (2) Yes if $\mathcal L=\O$ \\
 \end{tabular}
\newline{}
\newline Thus we have the following "oracle" version of Corollary 3.7.
\begin{cor}
We can decide if a basic quadruple is a uniform matroid or not in $O(|E|^2)$ by calling the zero locked number oracle one time.
\end{cor}
\begin{proof}
  Let $M=(\mathcal P, \mathcal S, \mathcal L, r)$ defined on $E$. By calling the zero locked number oracle, we can know if $\mathcal L$ is empty or not. According to the previous corollary, if $\mathcal L\neq \O$ then $M$ is not a uniform matroid. Otherwise, we have to check the following axioms only to decide if $M$ defines a matroid: (L1) ($O(1)$), (L2)-(L3) ($O(|E|^2)$), (LL7) ($O(1)$), (L8)-(L11) ($O(|E|)$), where the running time complexity for each of these axioms are indicated between brackets. Now by using Theorem 3.6 we can conclude as follows: (i) can be tested in $O(|E|)$, (ii)-(iii) can be tested in $O(|E|)$, and finally (iv)-(v) can be tested in $O(1)$.
\end{proof}

The difference between the two algorithms presented, respectively, in Corollary 3.7 and Corollary 3.9, is that, in the first one, the input can be exponential because of $|\mathcal L|$, but in the second one, the input has a size of at most $O(|E|)$ because we use the zero locked number oracle.

\section{Conclusion}

We have given a new system of axioms for defining matroids based mostly on locked subsets. We have deduced an improved algorithm for recognizing matroids given a basic quadruple. This algorithm becomes polynomial when recognizing uniform matroids. Future investigations can be improving the running time complexity of our algorithm.


\end{document}